\documentclass[11pt,a4paper]{article}

\usepackage{preprint}
\usepackage{amscd,mathtools}
\usepackage{booktabs,tabularx,array}
\usepackage{xurl}

\newcommand{\C}{\mathbb{C}}

\newcommand{\N}{\mathbb{N}}

\newcommand{\rank}{\operatorname{rank}}
\newcommand{\ran}{\operatorname{ran}}
\newcommand{\spec}{\operatorname{spec}}
\newcommand{\dist}{\operatorname{dist}}
\newcommand{\Cl}{\operatorname{Cl}}
\newcommand{\poly}{\operatorname{poly}}
\newcommand{\BQP}{\mathsf{BQP}}
\newcommand{\BQPone}{\mathsf{BQP}_{1}}
\newcommand{\SDQCone}{\mathsf{SDQC}_{1}}
\newcommand{\QMAone}{\mathsf{QMA}_{1}}
\newcommand{\CX}{\mathsf{CX}}
\newcommand{\CCX}{\mathsf{CCX}}
\newcommand{\cH}{\mathcal{H}}

\newcommand{\cV}{\mathcal{V}}

\newcommand{\ket}[1]{\lvert #1\rangle}
\newcommand{\bra}[1]{\langle #1\rvert}

\newcommand{\Xin}{X_{\mathrm{in}}}
\newcommand{\Xout}{X_{\mathrm{out}}}

\newtheorem{assumption}[theorem]{Assumption}
\title{Exact Kernel Transfer to Clique Complexes and the Hardness of
Normalized Persistence}
\author{Cheng Xin\\[0.4em]\small Department of Computer Science\\\small California State University, Fresno\\\small\texttt{cxin@mail.fresnostate.edu}}
\date{}
\AtBeginDocument{\hypersetup{
  pdftitle={Exact Kernel Transfer to Clique Complexes and the Hardness of Normalized Persistence},
  pdfauthor={Cheng Xin},
  pdfsubject={Preprint on exact normalized persistence in clique filtrations},
  pdfkeywords={clique homology, normalized persistence, spectral gap, quantum complexity}}}

\begin{document}
\raggedbottom
\maketitle

\begin{abstract}
For an inclusion of clique complexes \(X_1\subseteq X_2\), the normalized persistence in degree \(d\) is the fraction
\[
\frac{\rank[H_d(X_1)\to H_d(X_2)]}{\dim H_d(X_1)}
\]
of the degree-\(d\) holes of \(X_1\) that survive in \(X_2\). Estimating this fraction requires controlling the exact homology at both endpoints and the map induced by inclusion, even when both endpoint Laplacians have inverse-polynomial spectral gaps.
We prove that estimating the fraction to additive error \(1/24\) is hard for \(\BQPone^{G_2}\), the perfect-completeness class over the exact gate set \(G_2=\{X,\CX,\CCX,H\otimes H\}\), and hence for \(\BQPone\) over every finite gate set with entries in a cyclotomic field \(\mathbb Q(\zeta_{2^k})\), even for unweighted clique complexes.
The main tool is a finite-certificate kernel-transfer theorem. For a fixed palette of weighted clique gadgets satisfying finitely many exactly checkable local conditions, every unit chain \(x\) of the full geometric complex satisfies
\[
\dist(x,K)^2\le C\Bigl(t\lambda^2+\frac{\langle x,\Delta x\rangle}{g\lambda^{26}}\Bigr),
\]
where \(K\) is the embedded kernel of the simulated projector Hamiltonian, \(g\) its gap, \(t\) the number of gadgets, and \(\lambda\) the private vertex weight. Since \(\lambda\) is chosen independently of \(g\), the geometric Laplacian has exactly \(\dim K\) zero modes and a gap linear in \(g\) above them. Exact fillings identify the endpoint homology with a quotient \(V/W_A\) of the register cycle space, and nested term sets induce the natural quotient epimorphisms, so the persistent rank equals the later kernel dimension without any choice of compatible harmonic representatives. A fixed eight-dimensional label register turns a \(\BQPone^{G_2}\) verifier into instances with \(\beta_d(X_1)=8\) and normalized persistence \(3/4\) or \(1/8\), and an established common-copy blow-up transfers everything to unweighted graphs.

\end{abstract}

\section{Introduction}
\label{sec:intro}

Persistent homology records which topological features of a growing simplicial complex survive as the complex grows~\cite{edelsbrunner2002topological,zomorodian2005computing,carlsson2009topology}. For a two-step filtration \(X_1\subseteq X_2\) and a degree \(d\), the basic invariant is the persistent Betti number, the rank of the induced map \(H_d(X_1)\to H_d(X_2)\). Dividing it by \(\beta_d(X_1)\) gives a scale-free and directly interpretable statistic,
\begin{equation}
\label{eq:intro-ratio}
q_d(X_1,X_2)=\frac{\rank[H_d(X_1)\to H_d(X_2)]}{\beta_d(X_1)}\in[0,1],
\end{equation}
the fraction of the degree-\(d\) holes of \(X_1\) that are still present in \(X_2\). Recent work~\cite{lowe2026normalized} introduced this quantity under the name \emph{normalized harmonic persistence} and asked for its complexity on clique complexes of weighted graphs, the input model of quantum topological data analysis~\cite{lloyd2016quantum,hayakawa2022quantum,gyurik2022towards,berry2024analyzing,mcardle2026streamlined}. The motivation is that the known quantum algorithms for Betti and persistent Betti numbers return quantities normalized by the number of simplices~\cite{lloyd2016quantum,hayakawa2022quantum,schmidhuber2023limitations,apers2023simple}, whereas a fraction of holes is the statistic one would actually ask for. Under inverse-polynomial spectral-gap promises on the two endpoint Laplacians, that work proved that a low-energy relaxation of \(q_d\) is \(\mathsf{DQC}_1\)-hard and in \(\BQP\), conjectured that \(q_d\) itself is hard for a perfect-completeness variant \(\SDQCone\) of \(\mathsf{DQC}_1\) (Conjecture~1 of~\cite{lowe2026normalized}), and reduced this conjecture to a \emph{coefficient-sensitive kernel realization} of history Hamiltonians by clique complexes (Conjecture~2 of~\cite{lowe2026normalized}). That work identifies a kernel-preserving simulation of frustration-free Hamiltonians by combinatorial Laplacians as the missing ingredient and an important open question.

\paragraph{Why exact kernels are hard to realize.}
The exact-kernel primitive available in the literature is a clique-gadget construction~\cite{crichigno2024clique,king2026gapped}: a graph whose clique complex has one hole per computational basis state, together with one weighted gadget per local projector of a frustration-free Hamiltonian \(H\). The gadgets produce the desired zero modes, and the many-gadget spectral analysis~\cite{king2026gapped} proves a gap when the simulated \(H\) is positive definite. A persistence reduction needs more. It needs two Hamiltonians \(H_{\rm in}\) and \(H_{\rm out}=H_{\rm in}+H_{\rm rej}\) with known but generally large kernels, the exact multiplicity \(\dim\ker\Delta=\dim\ker H\) at both endpoints, an inverse-polynomial gap above the \emph{complete} kernel at both endpoints, and control of the inclusion-induced map. One must rule out every additional geometric zero mode and every small positive eigenvalue, including for chains supported almost entirely on private gadget vertices, with constants that do not depend on the exponentially large padded register.

\paragraph{Results.}
We carry out this realization for the history Hamiltonians of exact circuits over the gate set \(G_2=\{X,\CX,\CCX,H\otimes H\}\), and obtain, to our knowledge, the first unconditional hardness theorem for the exact fraction~\eqref{eq:intro-ratio} on clique complexes. Under the exact-circuit convention of~\cite{rudolph2026universal}, \(\BQPone^{G_2}\) is the class of promise problems decided with perfect completeness by uniform exact \(G_2\) circuits; by the exact simulation theorem in~\cite{rudolph2026universal} it contains \(\BQPone^{G}\) for every finite gate set \(G\) with entries in a cyclotomic field \(\mathbb Q(\zeta_{2^k})\), for instance Clifford\(+T\).

\begin{theorem}[Main theorem; \cref{thm:weighted-hardness,cor:unweighted-hardness}]
\label{thm:intro-main}
There is a deterministic polynomial-time reduction that maps an instance \(x\) of any \(\BQPone^{G_2}\) promise problem \(L\) to nested unweighted clique complexes \(\Xin\subseteq\Xout\) on a common vertex set, a degree \(d\), and rationals \(\gamma_{\rm in},\gamma_{\rm out}\ge1/\poly(|x|)\) such that \(\Delta_{\Xin,d}\) and \(\Delta_{\Xout,d}\) have no eigenvalue in \((0,\gamma_{\rm in})\) and \((0,\gamma_{\rm out})\) respectively, \(\beta_d(\Xin)=8\), and
\[
q_d(\Xin,\Xout)=\begin{cases}3/4,&x\in L,\\[2pt] 1/8,&x\notin L.\end{cases}
\]
Consequently, estimating \(q_d\) to additive error \(1/24\) with success probability \(2/3\), for weighted or unweighted clique complexes with inverse-polynomial endpoint gaps promised (Problem~9 of~\cite{lowe2026normalized}), is \(\BQPone^{G_2}\)-hard.
\end{theorem}

The reduction has three layers, each of independent interest.

\begin{theorem}[Kernel transfer; \cref{thm:transfer}, informal]
\label{thm:intro-transfer}
Fix a finite palette of weighted clique gadgets satisfying finitely many exactly checkable local conditions (\cref{ass:local-certificate}). Attach \(t\) gadgets with private vertex weight \(\lambda\) to the register, simulating \(H=\sum_j\Pi_j\) with \(H\succeq g(P_V-P_{K})\), \(K=\ker H\). Then every unit chain \(x\) of the full clique complex satisfies
\begin{equation}
\label{eq:intro-concentration}
\dist(x,K)^2\le C\Bigl[t\lambda^2+\frac{\langle x,\Delta x\rangle}{g\lambda^{\kappa}}\Bigr],\qquad\kappa=4m+2,
\end{equation}
where \(m\) bounds the number of active qubits of a gadget and \(C\) depends only on the palette. Hence, for \(\lambda\le c\min\{t^{-1},\eta/\sqrt t\}\) and \(E=c'\eta^2g\lambda^\kappa\),
\[
\dim\ker\Delta=\dim K,\qquad\spec(\Delta)\cap(0,E)=\varnothing,
\]
including the case \(K=0\).
\end{theorem}

\Cref{eq:intro-concentration} concerns arbitrary geometric chains, not only vectors of the simulated register, and its constants depend only on the palette. Its shape matters. The structural error \(t\lambda^2\) does not involve \(g\), so the gadget scale \(\lambda\) is chosen from the number of gadgets alone, independently of the logical gap, and the resulting geometric gap \(E\) is linear in \(g\). A direct use of a many-gadget inequality in which \(\lambda\) is proportional to \(g\) would instead produce a high power of \(g\). For the palette used here \(m=6\), and a conservative choice gives \(E=\Omega(\eta^{28}g/t^{26})\); the exponent is poor, but the bound is inverse polynomial.

\begin{theorem}[Quotient naturality; \cref{thm:quotient}]
\label{thm:intro-quotient}
Under the hypotheses of \cref{thm:intro-transfer}, \(H_d(X_A)\cong V/W_A\) with \(W_A=\sum_{j\in A}\ran\Pi_j=K_A^\perp\), and for \(A\subseteq B\) the inclusion \(X_A\subseteq X_B\) induces the natural quotient map \(V/W_A\twoheadrightarrow V/W_B\). In particular \(\rank[H_d(X_A)\to H_d(X_B)]=\dim K_B\).
\end{theorem}

\Cref{thm:intro-transfer,thm:intro-quotient} give exactly the two properties that the kernel-realization conjecture~\cite[Conjecture~2]{lowe2026normalized} requests, kernel fidelity with computable gaps and functoriality along the filtration, for the exact history Hamiltonians treated here; \cref{sec:related} states the precise comparison. Functoriality is obtained differently from the way it is formulated there. We never choose harmonic representatives compatibly across the filtration, which need not be possible since harmonic representatives rotate; instead, exact fillings identify each homology group with a quotient of the fixed register cycle space, and the induced map is then forced to be the natural quotient.

The third layer is an exact source. A three-bit maximally mixed label turns a \(\BQPone^{G_2}\) verifier with acceptance probability \(p_x\) into a verifier whose compressed acceptance operator on the label is \(\operatorname{diag}(1,p_x,p_x,p_x,p_x,p_x,0,0)\). Its perfect eigenspace has dimension \(6\) if \(p_x=1\) and \(1\) if \(p_x\le1/3\), which fixes \(\beta_d(\Xin)=8\) and the values \(3/4\) and \(1/8\). Finally, the common-copy blow-up~\cite{hayakawa2026unweighted}, applied with the same labeled copy blocks at both levels, transfers Betti numbers, persistent rank, and gaps to unweighted graphs.

\paragraph{Technical overview.}
Four ideas carry the proof of \cref{thm:intro-transfer}.

\emph{Finite certificates and the zero-weight limit.} In normalized coordinates the boundary pair of a gadget with private weight \(\lambda\) is affine in \(\lambda\), \(D(\lambda)=D_0+\lambda D_1\), and \(D_0\) is a fixed rational matrix. The certificate (\cref{ass:local-certificate}) asks for four exact facts about each palette member: the boundary intersection \(B_{d_a}(Y_a)\cap V_a\) equals the range of the intended rank-one projector (exact filling); \(\ker D_0=V_a\oplus Q_a\) for an explicit private subspace \(Q_a\); a projected private differential is coercive on \(Q_a\) at first order in \(\lambda\); and its retained outputs are private. All four are integer or modular linear algebra, and \cref{app:palette} records them for the palette of this paper.

\emph{Harmonic padding.} A certificate speaks about a constant-size active register, while the actual gadget is joined with the unaffected register factors, so the padded complex is exponentially large. The join Laplacian is a tensor sum over bidegrees, and the unaffected factors have harmonics only in top degree. We keep the local private data tensored with these outside harmonics and place every nonharmonic outside mode in a sector with a constant gap (\cref{lem:scaling-padding}). This is the step where the constants become independent of the padded dimension. A bulk consisting of all coordinates that contain a private vertex would not do, because deleting a weight-one outside vertex stays in that bulk with coefficient one.

\emph{Two energies.} The proof of \cref{eq:intro-concentration} separates leakage out of the register from logical energy. Leakage is controlled by the zero-weight floor, the private coercivity, and an interface bound \(\|[L_1\cdots L_t]\|^2=\|\sum_jL_jL_j^*\|\le Ct\lambda^2\) that needs no orthogonality between the register outputs of different gadgets. Logical energy is controlled by exact filling through \(\Delta^\uparrow\succeq c\lambda^\kappa H^{\rm emb}\). The two contributions carry different powers of \(\lambda\), which is what decouples \(\lambda\) from \(g\).

\emph{From estimate to exact multiplicity.} Every vector of the spectral subspace of \(\Delta\) below \(E\) projects nontrivially onto \(K\), so that subspace has dimension at most \(\dim K\). Exact fillings inject \(V/W_A\) into \(H_d(X_A)\), providing \(\dim K\) zero modes. Equality forces both the exact multiplicity and the absence of eigenvalues in \((0,E)\); no relative-acyclicity assumption is needed.

\paragraph{Scope.}
The hardness is gate-dependent, as in the known perfect-completeness hardness results for topological data analysis~\cite{crichigno2024clique,king2026gapped,gyurik2026provable}. We do not obtain ordinary \(\BQP\)-hardness, unrestricted \(\SDQCone\)-hardness, or gate-independent \(\BQPone\)-hardness, and complex phases and general integer states are outside the certified palette. The initial Betti number of our instances is the constant \(8\), produced by three mixed label qubits. The gap is a complexity promise rather than a practical estimate, and the common-copy blow-up is polynomial but large. \Cref{sec:limitations} discusses these points and the open problems they suggest.

\paragraph{Organization.}
\Cref{sec:prelim} fixes conventions and the problem. \Cref{sec:transfer} states the certificate model and the transfer theorem, \cref{sec:quotient} derives quotient naturality, \cref{sec:circuit} instantiates the history construction with a certified palette and proves \cref{thm:weighted-hardness}, and \cref{sec:unweighting} transfers the result to unweighted graphs. \Cref{sec:related} compares with prior work and \cref{sec:limitations} concludes. The appendices contain the complete transfer proof, the history-gap and source calculations, and the finite palette certificates.

\section{Preliminaries and problem statement}
\label{sec:prelim}

\paragraph{Weighted clique complexes.}
All chain spaces are finite-dimensional over \(\C\).  A vertex-weighted
simplicial complex has positive vertex weights \(w(v)\), and a simplex has
weight \(w(\sigma)=\prod_{v\in\sigma}w(v)\).  We use the normalized
coordinates of King and Kohler~\cite{king2026gapped}, in which the weighted
boundary of an oriented simplex removes a vertex \(v\) with coefficient
\(\pm w(v)\); in matrix form \(\partial_{w,k}=W_{k-1}^{-1}\partial_{1,k}W_k\),
where \(W_k\) is the diagonal matrix of \(k\)-simplex weights.  The weighted
complex is therefore conjugate to the unweighted one and has the same
homology.  We use
\[
\Delta_{X,d}=\partial_{X,d}^*\partial_{X,d}
+\partial_{X,d+1}\partial_{X,d+1}^*,\quad
Z_d(X)=\ker\partial_{X,d},\quad
B_d(X)=\ran\partial_{X,d+1}.
\]
Finite-dimensional Hodge theory~\cite{eckmann1944harmonische,horak2013spectra}
gives the orthogonal decomposition
\(C_d(X)=\cH_d(X)\oplus B_d(X)\oplus\ran\partial_{X,d}^*\) with
\(\cH_d(X)=\ker\Delta_{X,d}\) and \(Z_d(X)=\cH_d(X)\oplus B_d(X)\), and
identifies \(H_d(X)=Z_d(X)/B_d(X)\) with the harmonic space \(\cH_d(X)\).
For a positive semidefinite operator \(A\), write
\(\gamma_+(A)=\min(\spec(A)\setminus\{0\})\) for its smallest positive
eigenvalue, with \(\gamma_+(0)=+\infty\).
All certified boundary matrices and fillings below are rational.  Rank, kernel
dimension, image intersection, and filling validity are unchanged under scalar
extension \(\mathbb Q\subset\mathbb R\subset\mathbb C\), and the real
positive-semidefinite inequalities complexify verbatim.  We state homology
over \(\mathbb C\) while certifying the finite palette over \(\mathbb Q\).

\paragraph{Normalized persistence.}
For \(X_1\subseteq X_2\), define
\begin{equation}
\label{eq:true-ratio}
\beta_d(X_1\to X_2)
=\rank[H_d(X_1)\to H_d(X_2)],\qquad
q_d(X_1,X_2)
=\frac{\beta_d(X_1\to X_2)}{\beta_d(X_1)}.
\end{equation}
Lowe et al.~\cite{lowe2026normalized} write the numerator through harmonic
projectors, as the rank of \(P_{\cH_d(X_2)}\) restricted to \(\cH_d(X_1)\).
The two formulations agree.  For a cycle \(z\in Z_d(X_2)\), the vector
\(P_{\cH_d(X_2)}z\) is the harmonic representative of \([z]\), because
\(z-P_{\cH_d(X_2)}z\in B_d(X_2)\); hence \(z\mapsto P_{\cH_d(X_2)}z\) induces
the isomorphism \(H_d(X_2)\cong\cH_d(X_2)\).  Since
\(\cH_d(X_1)\subseteq Z_d(X_1)\subseteq Z_d(X_2)\) and
\(\cH_d(X_1)\cong H_d(X_1)\), the operator \(P_{\cH_d(X_2)}|_{\cH_d(X_1)}\) is
the induced map \(H_d(X_1)\to H_d(X_2)\) up to these identifications, and its
rank is \(\beta_d(X_1\to X_2)\).

\begin{definition}[Gapped normalized persistence]
\label{def:problem}
An instance of the weighted problem consists of explicit simple graphs
\(G_1\subseteq G_2\) on one labeled vertex set \(V\), one common positive
dyadic weight function \(w:V\to\mathbb Q_{>0}\), a degree \(d\ge1\), positive
rational lower bounds \(\gamma_1,\gamma_2\), and a positive rational accuracy
\(\varepsilon\), all encoded in binary and all at least \(1/\poly(|V|)\).
Put \(X_i=\Cl(G_i)\), using the
weighted-boundary convention above.  The promises are
\[
\beta_d(X_1)>0,
\qquad
\spec(\Delta_{X_i,d})\subseteq\{0\}\cup[\gamma_i,\infty)
\quad(i=1,2).
\]
The task is to output \(\widetilde q\) such that
\(\Pr[|\widetilde q-q_d(X_1,X_2)|\le\varepsilon]\ge2/3\).  Input size is the
total bit length of the graph encodings, weights, \(d\), the two gap bounds,
and \(\varepsilon\).  The unweighted problem is the restriction \(w\equiv1\).
Its fixed decision version distinguishes \(q_d\ge3/4\) from \(q_d\le1/8\).
\end{definition}
This is Problem~9 of Lowe et al.~\cite{lowe2026normalized}, with the same
promises, including their inverse-polynomial lower bounds on the gap and
accuracy parameters, and with the encoding conventions fixed.  The denominator is the
initial Betti number, not the number of simplices, and the statistic counts
exact zero modes rather than eigenvalues below a cutoff.

\paragraph{Register, logical projectors, and gadgets.}
The construction starts from a top-dimensional \emph{register} complex
\(R\), so \(C_{d+1}(R)=0\), and sets \(V=Z_d(R)=\cH_d(R)\).  Concretely,
\(R\) is the clique complex of the join of \(n\) \emph{bowtie} graphs, one
per logical qubit.  A bowtie is the wedge of two \(4\)-cycles; its clique
complex is one-dimensional and connected with \(\beta_1=2\), and the two
normalized oriented petal cycles are the computational basis
\(\ket0,\ket1\).  By the K\"unneth formula for joins, \(R\) has reduced
homology only in degree \(d=2n-1\), and \(V\cong(\C^2)^{\otimes n}\)
isometrically.  Each constraint \(j\) is a rank-one projector \(\pi_j\) on
the cycle space of a constant-size \emph{active} sub-register; after
tensoring with the identity of the unaffected factors it becomes an
orthogonal projector \(\Pi_j\) on \(V\), generally of higher rank.  For a
term set \(A\), put
\begin{equation}
\label{eq:logical-sum}
H_A=\sum_{j\in A}\Pi_j,\qquad
K_A=\ker H_A,\qquad
W_A=\sum_{j\in A}\ran\Pi_j=K_A^\perp ,
\end{equation}
and use a certified logical gap
\begin{equation}
\label{eq:logical-gap}
H_A\succeq g_A(P_V-P_{K_A}),\qquad 0<g_A\le1.
\end{equation}

Each constraint is realized by a weighted clique gadget: a graph containing
the active register, whose additional \emph{private} vertices have weight
\(\lambda\) and are joined to every vertex of the unaffected register
factors, while all register vertices have weight one.  Attaching the gadgets
of all \(j\in A\) to \(R\) gives the clique complex \(X_A\), and \(\Delta_A\)
denotes its full Laplacian on all degree-\(d\) chains.  Thus \(H_A\) acts on
\(V\), while \(\Delta_A\) acts on \(C_d(X_A)\supseteq V\).

\begin{definition}[Independent private interiors]
\label{def:independent}
The gadgets of a term set \(A\) have \emph{independent private interiors}
if private vertices of distinct gadgets are never adjacent.  Every simplex
of \(X_A\) that contains a private vertex then contains private vertices of
exactly one gadget, and
\(C_d(X_A)=C_d(R)\oplus\bigoplus_{j\in A}C^{\rm priv}_{d,j}\) as coordinate
spaces.
\end{definition}
All constructions in this paper satisfy \cref{def:independent}: distinct
gadgets receive disjoint sets of new vertices, and no edge is added between
private vertices of distinct gadgets.

\paragraph{Circuit source.}
For the circuit application, \(J\) inserts the fixed clean string around a
maximally mixed logical input, \(U\) is an exact circuit, and
\[
M=J^*U^*P_{\rm acc}UJ,\qquad S=\ker(I-M).
\]
The initial history kernel will have the full input dimension, while the
final history kernel will be isomorphic to the perfect accepting space
\(S\).

\section{Finite-certificate transfer to the complete geometric kernel}
\label{sec:transfer}

We isolate finite local conditions under which a clique-gadget construction
controls every chain of the full geometric complex, including chains
supported mostly outside the simulated register.  The statement does not use
convergence of a spectral projector whose rank grows with the simulated
circuit.

\subsection{The certificate model}

Consider one active gadget before padding.  Its weight-one register is \(R_a\),
its target cycle space is \(V_a=Z_{d_a}(R_a)\), and every private vertex has
weight \(\lambda\).  In normalized chain coordinates put
\[
D_a(\lambda)=
\begin{bmatrix}\partial_{d_a}(\lambda)\\
\partial_{d_a+1}(\lambda)^*\end{bmatrix}
=D_{a,0}+\lambda D_{a,1},
\]
so that \(\Delta_{Y_a,d_a}(\lambda)=D_a(\lambda)^*D_a(\lambda)\).  The
decomposition is affine in \(\lambda\) because every entry of the normalized
boundary is \(0\), \(\pm1\) (a register vertex removed) or \(\pm\lambda\) (a
private vertex removed).  The value at \(\lambda=0\) is a matrix limit, not
a singular change of basis; \(D_{a,0}\) is the boundary pair of a rational
chain complex, since \(\partial(\lambda)^2=0\) for all \(\lambda\) forces
\(\partial_0^2=0\).

\begin{assumption}[Finite local certificate]
\label{ass:local-certificate}
For every member \(a\) of a fixed finite palette there are an orthogonal
subspace \(Q_a\) contained in the span of private simplex coordinates, an
output projector \(R_a^{\rm out}\), and a rank-one
projector \(\pi_a\) on \(V_a\) such that:
\begin{enumerate}[(i)]
\item \(B_{d_a}(Y_a)\cap V_a=\ran\pi_a\) for every \(\lambda>0\);
\item \(\ker D_{a,0}=V_a\oplus Q_a\), orthogonally;
\item for \(T_a(\lambda)=R_a^{\rm out}D_a(\lambda)\),
\[
T_a(\lambda)=\lambda T_{a,1},\qquad T_{a,1}V_a=0,\qquad
\|T_{a,1}q\|\ge b\|q\|\quad(q\in Q_a);
\]
\item every retained output coordinate contains a private vertex.
\end{enumerate}
The constant \(b>0\), the positive singular-value floor of \(D_{a,0}\), and
all matrix norms are uniform over the palette.
\end{assumption}

All data in \cref{ass:local-certificate} are finite integer or rational linear
algebra.  Exact fillings prove one inclusion in (i); matching boundary ranks
and surviving register cosets, or the guard connecting-map argument, prove the
reverse inclusion.  Condition (i) does not depend on \(\lambda\): the
weighted boundary space is \(W_{d_a}^{-1}B_{d_a}(Y_a;1)\) and \(W_{d_a}\)
is the identity on register coordinates.  Modular rank bounds with matching
rational kernels certify (ii), and a nonzero Gram determinant certifies
(iii).  The subspace \(Q_a\) need not be spanned by simplex coordinates; the
selected-cycle guards of \cref{app:palette} produce \(Q_a\) of the form
\(Q\otimes\C\gamma\).  The explicit palette used later is recorded in
\cref{app:palette}.

\begin{lemma}[Weight scaling and harmonic padding]
\label{lem:scaling-padding}
Let \(0<\lambda\le1\), let \(d_a=2m_a-1\) with \(m_a\le m\), and put
\(\kappa:=4m+2\).  There is a palette constant \(c>0\) such that
\begin{equation}
\label{eq:local-gap}
\gamma_+\!\left(\Delta_{Y_a,d_a}(\lambda)\right)
\ge c\lambda^\kappa .
\end{equation}
Let \(Y_a\) be joined with the unaffected register factors, whose join
\(R_{\rm out}\) has reduced homology only in its top degree, with top
harmonic space \(\cH_{\rm out}\), and let \(d\) be the global target degree.
Then the zero-weight kernel of the padded gadget in degree \(d\) is
\[
V\oplus\bigl(Q_a\otimes\cH_{\rm out}\bigr),
\]
its orthogonal complement has a uniform positive zero-weight floor, the
projected pair \(T_a(\lambda)\otimes P_{\cH_{\rm out}}\) satisfies (iii)
with the same constant \(b\), and the padded exact boundary intersection
with \(V\) is the range of
\(\Pi_a=\pi_a\otimes I_{\cH_{\rm out}}\).
\end{lemma}

\begin{proof}
See \cref{app:local-padding}.
\end{proof}

The harmonic restriction in this lemma is essential.  If the padded bulk
contained every coordinate with a local private vertex, deleting a weight-one
outside vertex could remain in that bulk with coefficient one.  We instead
keep local private coordinates tensored with outside harmonics and place every
nonharmonic outside factor in the separately gapped sector.

\subsection{Global theorem}

Attach \(t_A\) padded gadgets to the common top-dimensional register \(R\)
with independent private interiors (\cref{def:independent}).  Let \(L_j\) be
the block of \(\partial_{X_A,d}\) that maps the private \(d\)-chains of
gadget \(j\) to the register \((d-1)\)-chains.  A private \(d\)-simplex has a
register face only if it contains exactly one private vertex, and removing
that vertex gives coefficient \(\pm\lambda\); hence every column of
\(\lambda^{-1}L_j\) has at most one nonzero entry and every row has at most
\(p_j\) nonzero entries, where \(p_j\) is the number of private vertices of
gadget \(j\).  Consequently the finite interface bound
\begin{equation}
\label{eq:interface}
L_jL_j^*\preceq N_j\lambda^2I,\qquad N_j\le N_*,
\end{equation}
holds with \(N_j\le p_j\), so \(N_*\) is a palette constant.
The input blocks are orthogonal, while the outputs may overlap.  Hence, for
\(L=[L_1\ \cdots\ L_{t_A}]\),
\[
LL^*=\sum_jL_jL_j^*\preceq N_*t_A\lambda^2I.
\]

\begin{theorem}[Finite-certificate degenerate-kernel transfer]
\label{thm:transfer}
Fix a finite palette satisfying \cref{ass:local-certificate}, with active
degrees \(d_a=2m_a-1\), \(m_a\le m\), and let \(\kappa=4m+2\).  Let
\(\mathcal A\) be a finite family of term sets whose gadgets are padded
palette members attached to the register \(R\) of \cref{sec:prelim} with
independent private interiors, all with the same private weight \(\lambda\).
Let \(t_{\max}=\max\{1,\max_{A\in\mathcal A}|A|\}\) and suppose
\[
H_A=\sum_{j\in A}\Pi_j\succeq
g_A(P_V-P_{K_A}),\qquad K_A=\ker H_A,\qquad 0<g_A\le1
\]
for every \(A\in\mathcal A\).
There are palette constants \(c_0,C>0\) such that, whenever
\(0<\lambda\le\min\{1,c_0/t_{\max}\}\), every unit \(x\in C_d(X_A)\)
satisfies
\begin{equation}
\label{eq:transfer-bound}
\boxed{\;
\|(I-P_{K_A})x\|^2
\le C\left[t_A\lambda^2+
\frac{\langle x,\Delta_Ax\rangle}{g_A\lambda^\kappa}\right].\;}
\end{equation}
Here \(P_{K_A}\) is the projector onto the embedded logical kernel, extended
by zero outside the register.

For \(0<\eta<1\), choose additionally
\begin{equation}
\label{eq:lambda-E-choice}
\lambda\le c_1\min\{t_{\max}^{-1},\eta/\sqrt{t_{\max}}\},
\qquad E_A\le c_2\eta^2g_A\lambda^\kappa,
\end{equation}
with palette constants \(c_1,c_2>0\).  Then
\begin{equation}
\label{eq:kernel-gap-conclusion}
\dim\ker\Delta_A=\dim K_A,\qquad
\spec(\Delta_A)\cap(0,E_A)=\varnothing .
\end{equation}
This includes \(K_A=0\), in which case \(\Delta_A\succeq E_AI\).
\end{theorem}

\paragraph{Proof architecture.}
The full calculation is in \cref{app:transfer-proof}.  The shared-output
identity above and the private projected pairs first give the geometric
leakage estimate
\begin{equation}
\label{eq:geometric-leakage}
\|(I-P_V)x\|^2
\le C\left[t_A\lambda^2+
(t_A+\lambda^{-2})\langle x,\Delta_Ax\rangle\right].
\end{equation}
Exact filling and \cref{lem:scaling-padding} separately give
\begin{equation}
\label{eq:logical-domination}
\Delta_A^\uparrow\succeq c\lambda^\kappa H_A^{\rm emb},
\end{equation}
where \(\Delta_A^\uparrow=\partial_{d+1}\partial_{d+1}^*\) and
\(H_A^{\rm emb}\) is \(H_A\) extended by zero outside \(V\).
Combining these inequalities with the logical gap proves
\cref{eq:transfer-bound}.  The bound applies to every vector in a low-energy
spectral subspace.  Projection onto \(K_A\) is therefore injective on that
subspace, so its dimension is at most \(\dim K_A\).  Independent exact
fillings supply \(\dim K_A\) zero modes through \(V/W_A\); equality and the
whole-positive-spectrum gap follow.  No relative-acyclicity assumption is
used.

For the explicit history palette \(m\le6\), so \(\kappa=26\).  Taking a
dyadic \(\lambda=\Theta(\eta/t_{\max})\) gives the conservative scale
\begin{equation}
\label{eq:weighted-gap}
E_A=\Omega\!\left(\frac{\eta^{28}g_A}{t_{\max}^{26}}\right).
\end{equation}
The exponent is not claimed to be optimal.

\section{Filtered quotient naturality}
\label{sec:quotient}

The transfer theorem determines the dimension of every endpoint kernel. Exact filling determines the maps between those kernels at the homology level. This separation is useful: spectral concentration proves that there are no extra classes, while the quotient model computes persistence without comparing harmonic representatives.

\begin{lemma}[Independent filling]
\label{lem:independent-filling}
For every term set \(A\),
\begin{equation}
\label{eq:boundary-intersection}
B_d(X_A)\cap V=W_A
=\sum_{j\in A}\ran\Pi_j.
\end{equation}
\end{lemma}

\begin{proof}
Every vector in \(\ran\Pi_j\) has a filling in gadget \(j\), so \(W_A\subseteq B_d(X_A)\cap V\). Conversely, write a degree-\((d+1)\) chain as a sum of its gadget components. If its total boundary lies in \(V\), then its boundary has zero coordinate on every private \(d\)-simplex. Private supports of distinct gadgets are disjoint, so this vanishing holds separately for each gadget component. The remaining register boundary of component \(j\) lies in \(B_d(Y_j)\cap V=\ran\Pi_j\). Summing over \(j\) proves the reverse inclusion.
\end{proof}

No relative-acyclicity hypothesis appears in \cref{lem:independent-filling}. Extra geometric homology outside \(V\) is excluded later by the all-chain concentration theorem, not by the boundary decomposition alone.

\begin{theorem}[Natural quotient realization]
\label{thm:quotient}
Under the hypotheses and parameter choice of \cref{thm:transfer}, the map
\begin{equation}
\label{eq:qA}
q_A:V/W_A\longrightarrow H_d(X_A),
\qquad
q_A(v+W_A)=[v],
\end{equation}
is an isomorphism. For \(A\subseteq B\), these isomorphisms make the diagram
\[
\begin{CD}
V/W_A @>{\rho_{AB}}>> V/W_B\\
@V{q_A}V{\cong}V @VV{\cong}V{q_B}\\
H_d(X_A) @>{(\iota_{AB})_*}>> H_d(X_B)
\end{CD}
\]
commute, where \(\rho_{AB}\) is the natural quotient map. Consequently,
\begin{equation}
\label{eq:persistent-rank}
\rank(\iota_{AB})_*
=\dim(V/W_B)
=\dim K_B.
\end{equation}
\end{theorem}

\begin{proof}
See \cref{app:quotient-proof}.
\end{proof}

For a chain \(A_0\subseteq A_1\subseteq\cdots\subseteq A_s\), every arrow in degree \(d\) is an epimorphism. The constructed barcode therefore has deaths but no new births in this degree. More generally, for any \(i\le j\),
\[
\rank\bigl[H_d(X_{A_i})\to H_d(X_{A_j})\bigr]
=\dim K_{A_j}.
\]
Harmonic representatives may rotate between filtration levels; the exact
statement concerns the algebraic homology maps.

Taking \(A_{\rm in}\) to be the clock, propagation, and clean-input terms of a history Hamiltonian, and taking \(A_{\rm out}\) to add the final rejection term, gives
\begin{equation}
\label{eq:history-ratio}
q_d(\Xin,\Xout)
=\frac{\dim K_{\rm out}}{\dim K_{\rm in}}.
\end{equation}
This identity is the bridge from exact logical perfect-space fractions to normalized persistence.

\section{Exact histories, a fixed clique palette, and hardness}
\label{sec:circuit}

Let
\[
G_2=\{X,\CX,\CCX,H\otimes H\}.
\]
Following Rudolph's gate-dependent definition
\cite[Definition~2.2]{rudolph2026universal},\footnote{Definition, theorem,
and appendix numbers for~\cite{rudolph2026universal} refer to the full
version, arXiv:2411.02681v2.} \(\BQPone^{G_2}\) consists of
promise problems decided by a polynomial-time uniform family of exact
\(G_2\) circuits on an all-zero input, with one computational-basis output
bit, perfect completeness, and inverse-polynomial rejection probability in
the no case.  Rudolph's exact error reduction
\cite[Theorem~2.3]{rudolph2026universal} preserves perfect completeness and
stays in any gate set containing \(G_2\); we apply it first so that soundness
is at most \(1/3\).  The superscript is essential: no approximate compilation
or gate-independent identification is used below.

\subsection{Rational histories and the certified palette}

Write \(K=\sqrt2H\).  Replace each \(H\otimes H\) transition by a gain edge
\(K\) on one work bit followed by a loss edge \(K^{-1}=H/\sqrt2\) on the
other.  Their composite is exactly \(H\otimes H\), and each edge is enforced
by two orthonormal rational three-term vectors.  The vectors and their direct
calculation appear in \cref{app:rational-split}.  This construction is due to
Rudolph~\cite[Appendix~D.1]{rudolph2026universal}.

The propagation constraints for \(X,\CX,\CCX\) decompose, on their
constant-size active support, into two-time vectors
\[
(\ket{0,z}-\ket{1,Uz})/\sqrt2.
\]
Clock, input, and output penalties are basis projectors.  Unaffected work
qubits are identities and are implemented by the harmonic padding of
\cref{lem:scaling-padding}; they are not expanded into exponentially many
rank-one terms.

\begin{lemma}[Certified fixed palette]
\label{lem:palette}
Every local constraint in the split \(G_2\) history is implemented by a
member of one fixed family of weighted clique gadgets of total support
locality at most six.  Every member satisfies
\cref{ass:local-certificate,eq:interface}, with uniform constants, and the
graph construction is polynomial in the circuit size.
\end{lemma}

\begin{proof}
See \cref{app:palette}.  It records the finite integer certificates, exact
relabelings, selected-cycle guard closure, and polynomial-size bookkeeping.
\end{proof}

\subsection{Complete history kernels}

Let \(L\) be the number of legal clock positions after splitting the
Hadamard transitions.  Put \(s_t=\sqrt2\) immediately after a gain edge and
\(s_t=1\) otherwise.  If \(V_t\) is the unitary prefix after removing these
scalars, then
\begin{equation}
\label{eq:history-isometry}
\cV v=Z^{-1/2}\sum_{t=0}^{L-1}s_t\ket tV_tv,\qquad
Z=\sum_ts_t^2,\qquad L\le Z\le2L.
\end{equation}
Every gain edge is immediately followed by its matching loss edge, with no
gate between them, so \(s_0=s_{L-1}=1\).

\begin{proposition}[History kernels and gaps]
\label{prop:history}
Let \(C\) be the valid clean-input space, \(D=\dim C>0\), and let
\(M=J^*U^*P_{\rm acc}UJ\) be the compressed acceptance operator.  Put
\(S=\ker(I-M)\).  The initial history Hamiltonian has complete kernel
\(\cV C\), while adding the final rejection projector gives complete kernel
\(\cV S\).  If \(M|_{S^\perp}\preceq I/3\), then
\begin{equation}
\label{eq:history-gaps}
g_{\rm in}\ge\frac1{8L^2},\qquad
g_{\rm out}\ge\frac1{3L(8L^2+1)}\ge\frac1{27L^3}.
\end{equation}
The second statement includes \(S=0\).
\end{proposition}

The proof in \cref{app:source-details} uses a weighted path Poincar\'e
inequality.  The output compression is exactly
\begin{equation}
\label{eq:output-compression}
\cV^*H_{\rm rej}\cV=(I-M)/Z,
\end{equation}
and a two-positive-operator estimate turns this angle penalty into the second
gap in \cref{eq:history-gaps}.  Combining
\cref{thm:transfer}, \cref{thm:quotient}, \cref{lem:palette}, and
\cref{prop:history} yields weighted
clique complexes \(\Xin\subseteq\Xout\) with
\begin{equation}
\label{eq:generic-circuit-transfer}
\beta_d(\Xin)=D,\qquad
\beta_d(\Xout)=\beta_d(\Xin\to\Xout)=\dim S,
\end{equation}
and inverse-polynomial gaps above both complete kernels.

\subsection{A fixed-eight source}

Given a \(\BQPone^{G_2}\) verifier \(Q_x\) with acceptance probability \(p_x\),
add a maximally mixed label \(z\in\{0,1\}^3\), run \(Q_x\) unconditionally,
and accept according to
\[
[z=000]\ \lor\
\bigl(q\land[z\in\{001,010,011,100,101\}]\bigr),
\]
where \(q\) is the output of \(Q_x\).  The reversible predicate uses only
\(X,\CX,\CCX\), preserves \(z\), and uncomputes its scratch.  Its compressed
acceptance operator on the label is exactly
\begin{equation}
\label{eq:eight-label}
M_x=\operatorname{diag}(1,p_x,p_x,p_x,p_x,p_x,0,0).
\end{equation}
Thus the perfect eigenspace has dimension six when \(p_x=1\) and dimension
one when \(p_x\le1/3\); every off-perfect eigenvalue in the latter case is at
most \(1/3\).

\begin{theorem}[Fixed-eight hardness]
\label{thm:weighted-hardness}
There is a deterministic polynomial-time many-one reduction from
\(\BQPone^{G_2}\) to the decision version of weighted gapped normalized
persistence (\cref{def:problem}).  Every output
satisfies
\[
\beta_d(\Xin)=8,\qquad
\beta_d(\Xin\to\Xout)=
\begin{cases}6,&x\in L,\\1,&x\notin L,\end{cases}
\]
and both endpoint degree-\(d\) Laplacians have inverse-polynomial gaps above
their complete kernels.  Hence the normalized values are \(3/4\) and \(1/8\).
Additive error \(1/24\), with decision threshold \(7/16\), distinguishes the
two cases with success probability at least \(2/3\).
\end{theorem}

\begin{proof}
Apply \cref{eq:generic-circuit-transfer} to \cref{eq:eight-label}.  The two
ratios differ by \(5/8\), and \(1/24<5/16\).  The logical gaps in
\cref{eq:history-gaps} are explicit rationals.  The fixed palette supplies
hard-coded positive rational lower bounds for the constants in
\cref{thm:transfer}.  Fix \(\eta_0=1/10\); from the polynomially bounded
\(t_{\max}\), choose the first integer \(b\) for which
\(\lambda=2^{-b}\) satisfies both inequalities in
\cref{eq:lambda-E-choice}.  For each endpoint let \(s_i\) be the least
integer with \(2^{-s_i}\le\frac12c_2\eta_0^2g_i\lambda^{26}\), and output
\(E_i=2^{-s_i}\), which satisfies \cref{eq:lambda-E-choice}.  Thus
\(b=O(\log t_{\max})\) and \(s_i=O(\log t_{\max}+\log L)\), so the binary
encodings of \(\lambda\) and \(E_i\) have polynomial length.  If the history
uses \(n\) register qubits (clock and work bits), the register has \(7n\)
vertices and \(O(n^2)\) edges, and each gadget adds a constant number of
private vertices together with \(O(n)\) edges to the unaffected register
factors; so the weighted graph has \(O(n+t_{\max})\) vertices and
\(O(n^2+t_{\max}n)\) explicitly generated edges.  Hence the map
outputs the two graphs, their common weight function, \(d,E_1,E_2\), and the
fixed decision threshold in deterministic polynomial time.
\end{proof}

\begin{remark}[Integer weights]
\label{rem:integer-weights}
Multiplying every vertex weight by \(2^{b}\) gives an equivalent instance
with integer weights in \(\{1,2^{b}\}\): the normalized boundary
\(W_{k-1}^{-1}\partial W_k\) scales by \(2^{b}\), so the Laplacian and both
gap bounds scale by \(2^{2b}\), while homology is unchanged.  The weighted
instances therefore also satisfy the polynomially bounded weight convention
of~\cite{lowe2026normalized}.
\end{remark}

\begin{remark}[Other gate sets]
\label{rem:gate-sets}
By Rudolph's exact simulation theorem~\cite[Theorem~3.4]{rudolph2026universal},
\(\BQPone^{G}\subseteq\BQPone^{G_2}\) for every finite gate set \(G\) whose
entries lie in a cyclotomic field \(\mathbb Q(\zeta_{2^k})\), \(k\in\N\);
Clifford\(+T\) is an example.  Composing that simulation with
\cref{thm:weighted-hardness} gives \(\BQPone^{G}\)-hardness for each such
\(G\) separately.  This is a gate-family-relative statement, not
gate-independent \(\BQPone\)-hardness.
\end{remark}

The exact trace-versus-perfect-space obstruction and the mixed-spectator
extension to individual real Hadamards are recorded in
\cref{app:source-boundaries}.  Neither is needed for
\cref{thm:weighted-hardness}.

\section{A natural unweighted corollary}
\label{sec:unweighting}

The weighted construction uses only weights \(1\) and a dyadic
inverse-polynomial \(\lambda\).  Hayakawa's common-copy construction replaces
each vertex \(v\) by a clique of \(f_v\) labeled copies and each base edge by
all cross-block edges~\cite{hayakawa2026unweighted}.  Choose one common
\(F=\lambda^{-2}\) and
\[
f_v=Fw(v)^2.
\]
Register vertices then have \(F\) copies and private vertices have one.

In degree \(k\), let \(U_{i,k}\) map a normalized base simplex to the
normalized average of its one-copy-per-block lifts.  Lemma~4.2 and
Theorem~4.3 of Hayakawa~\cite{hayakawa2026unweighted} give the symmetric
identity and asymmetric gap over real chains; scalar extension gives the same
statements over \(\mathbb C\):
\begin{equation}
\label{eq:unweight-operator}
\widehat\Delta_{i,k}U_{i,k}=FU_{i,k}\Delta_{i,k},
\qquad
\widehat\Delta_{i,k}|_{\ran(U_{i,k})^\perp}\succeq
\min_v f_v\,I.
\end{equation}
Use the same labeled copy block for every shared vertex at both filtration
levels.  With consistent orientations, simplicial inclusion then satisfies
\begin{equation}
\label{eq:unweight-natural}
\widehat J_{ij}U_{i,d}=U_{j,d}J_{ij}.
\end{equation}
There is no extra normalization factor because an inclusion preserves the
simplex degree, weights, labels, and copy set.

\begin{corollary}[Unweighted normalized persistence]
\label{cor:unweighted-hardness}
The reduction in \cref{thm:weighted-hardness} maps in polynomial time to
nested unweighted clique complexes with the same endpoint Betti numbers,
persistent rank, and normalized values \(3/4\) and \(1/8\).  Both unweighted
endpoint Laplacians retain inverse-polynomial gaps above their complete
kernels.
\end{corollary}

\begin{proof}
The asymmetric bound in \cref{eq:unweight-operator} excludes asymmetric
harmonic vectors, so
\[
\ker\widehat\Delta_{i,d}=U_{i,d}\ker\Delta_{i,d}.
\]
Together with \cref{eq:unweight-natural}, this makes every inclusion-induced
map on harmonic spaces unitarily equivalent to its weighted counterpart.
Endpoint nullities and induced ranks are therefore preserved.  If \(E_i\)
lower-bounds the positive weighted spectrum, then
\[
\gamma_+(\widehat\Delta_{i,d})
\ge\min\{FE_i,\min_v f_v\}=\min\{FE_i,1\}.
\]
By \cref{eq:weighted-gap}, \(FE_i=
\Omega(\eta^{26}g_i/t_{\max}^{24})\).  Since \(F\) is polynomial, so is the
blown-up graph.  More explicitly, if the weighted graph has
\(O(n+t_{\max})\) vertices and \(O(n^2+t_{\max}n)\) edges, a safe explicit
bound after blow-up is \(O(F(n+t_{\max}))\) vertices and
\(O(F^2(n^2+t_{\max}n))\) edges; labels and edges are generated by nested
loops in polynomial time.  Future private vertices may be present initially as isolated
vertices; the target degree is positive, so they add no target-degree chain.
\end{proof}

The operator decomposition is Hayakawa's; the filtration-specific step is
the commuting identity \cref{eq:unweight-natural}.  We do not count the
unweighting mechanism as a separate contribution.

\section{Relation to prior work}
\label{sec:related}

\paragraph{Normalized persistence.}
Lowe et al.~\cite{lowe2026normalized} define normalized harmonic persistence
(their Problem~9, our \cref{def:problem}), prove that a low-energy relaxation
of it is \(\mathsf{DQC}_1\)-hard and in \(\BQP\), prove \(\SDQCone\)-hardness
of the local-Hamiltonian analogue, and conjecture \(\SDQCone\)-hardness of
Problem~9 (their Conjecture~1).  Their Conjecture~2 asks for a polynomial-time
map from history Hamiltonians \(H_1\) and \(H_2=H_1+\lambda H_{\rm out}\) to
weighted clique complexes \(X_1\subseteq X_2\) with two properties:
kernel fidelity and gap, \(\beta_d(X_i)=\dim\ker H_i\) together with
computable inverse-polynomial gap bounds; and functoriality,
\(\beta_d(X_1\to X_2)=\dim\ran(P_{\ker H_2}|_{\ker H_1})\).  Their Lemma~11
shows that Conjecture~2 implies Conjecture~1.  \Cref{thm:transfer,thm:quotient}
establish both properties for the exact \(G_2\) history Hamiltonians of
\cref{sec:circuit}, with \(H_{\rm out}\) at unit coefficient rather than at a
small coefficient \(\lambda\), and with the induced map computed through the
quotient model rather than through compatible harmonic isometries as
formulated there.  Two differences prevent us from concluding Conjecture~1
itself.  First, our source is \(\BQPone^{G_2}\) with three mixed label
qubits rather than an \(\SDQCone\) verifier with a maximally mixed register.
The \(\SDQCone\) promise constrains the acceptance probability of the
maximally mixed state, and a trace promise does not in general determine the
dimension of the perfectly accepted subspace, which is what \(q_d\) measures
(\cref{app:source-boundaries}).  Second, the certified palette is real and
exact, so verifiers over an arbitrary gate set are not covered.  Their
containment result (Lemma~7) needs additional access assumptions, namely a
preparation circuit for the uniform mixture over \(\ker\Delta_{1,d}\) and an
inverse-polynomial lower bound on the nonzero principal angles between the
two harmonic spaces.  We do not establish these for our instances, and
containment of the promise problem of \cref{def:problem} in \(\BQP\) remains
open.

\paragraph{Clique homology gadgets.}
Crichigno and Kohler~\cite{crichigno2024clique} give a parsimonious reduction
in which satisfying states correspond one-to-one to holes, so exact
multiplicity by itself is not new; their instances carry no gap promise.
King and Kohler~\cite{king2026gapped} introduce the weighted gadgets and the
many-gadget spectral analysis behind gapped clique homology, proving a gap
when the logical Hamiltonian is positive definite, with a gadget weight
proportional to the logical gap; Hayakawa~\cite{hayakawa2026unweighted}
removes the weights under the same positive-definite promise.  Our transfer
theorem is formulated from finite zero-weight certificates, applies to
arbitrary geometric chains and to degenerate kernels, and decouples
\(\lambda\) from \(g\).  Its padding step differs from the coordinate-bulk
formulation displayed in the arXiv version of King and
Kohler~\cite[Claim~10.2]{king2023gappedv2}: there the bulk consists of all
coordinates containing a private vertex, which in the padded complex also
contains boundary terms obtained by deleting a weight-one outside vertex,
whereas we keep local private coordinates tensored with outside harmonics
in every bidegree (\cref{lem:scaling-padding}).  We also work with the fixed
local matrices before tensoring with the outside identity, so that no
operator bound on a projector of growing rank has to be inferred from
basiswise perturbation estimates.  These are differences of proof technique,
and we make no claim about the presentation in the published
version~\cite{king2026gapped}.

\paragraph{Persistence hardness.}
Gyurik et al.~\cite{gyurik2026provable} prove \(\BQPone\)-hardness, over the
gate set \(\{\mathrm{CNOT},U_{\rm Pyth}\}\), of deciding whether a
\emph{given} harmonic representative of \(X_1\) survives in \(X_2\), with a
gap promise on the final Laplacian.  Normalized persistence asks instead for
the fraction of the full initial homology that survives, with no
representative given, so a hardness proof must control the complete kernels
at both endpoints.  Persistent Laplacians~\cite{memoli2022persistent} give a
spectral formulation of persistent Betti numbers; we use ordinary endpoint
Laplacians for the promise and the quotient model for the map.

\paragraph{Sources and gate sets.}
Rudolph~\cite{rudolph2026universal} defines the gate-dependent classes,
proves exact error reduction inside gate sets containing \(G_2\), shows
\(\BQPone^{G}\subseteq\BQPone^{G_2}\) for every finite \(G\) with entries in
some \(\mathbb Q(\zeta_{2^k})\), and supplies the gain/loss split of
\(H\otimes H\) with rational three-term constraint vectors and the base
clique gadgets used in our palette.  Cade and
Crichigno~\cite{cade2024complexity} relate cohomology problems to
\(\QMAone\) and \(\mathsf{DQC}_1\), and Gyurik, Cade, and
Dunjko~\cite{gyurik2022towards} connected quantum topological data analysis
to \(\mathsf{DQC}_1\)~\cite{knill1998power}.  The history-state construction
with a unary clock and neighbor guards is
Kitaev's~\cite{kitaev2002classical,kempe2006complexity}, with the rational
amplitude bookkeeping of \cref{sec:circuit}.

\paragraph{Unweighting.}
Hayakawa's symmetric/asymmetric decomposition of the common-copy
blow-up~\cite[Lemma~4.2 and Theorem~4.3]{hayakawa2026unweighted} is used as a
black box.  Our only addition is the observation that identical labeled copy
blocks at both filtration levels make the symmetric isometries commute with
inclusion, \cref{eq:unweight-natural}.  We do not count this as a separate
contribution.

\section{Limitations and open problems}
\label{sec:limitations}

The construction has four substantive limits.  Its exact palette is finite
and real: arbitrary integer states and complex phases are outside the
theorem, so the hardness source is the gate-dependent class
\(\BQPone^{G_2}\) rather than ordinary \(\BQP\), unrestricted \(\SDQCone\),
or a gate-independent perfect-completeness class.  The clean source has
\(\beta_d(\Xin)=8\); ignored mixed qubits replicate this denominator
exactly, but an intrinsically growing family of distinct initial classes
would require another exact source.  The gap is a complexity promise rather
than a practical estimate: at locality six it scales conservatively as
\(\Omega(\eta^{28}g/t^{26})\) with \(g=\Omega(L^{-3})\), and the common-copy
blow-up multiplies the vertex count by \(F=\lambda^{-2}\).  Finally, the
hardness is for the promise problem with endpoint gaps; whether our instances
also satisfy the access assumptions under which Lowe et al.\ place the
problem in \(\BQP\) is not established.

Two questions would strengthen the complexity classification.
\begin{enumerate}
\item Is normalized harmonic persistence \(\SDQCone\)-hard, as conjectured
in~\cite{lowe2026normalized}?  The missing steps are a source whose
perfect-space fraction is separated by the \(\SDQCone\) trace promise, and a
certified palette for the corresponding gate set.
\item Do the constructed filtrations satisfy an inverse-polynomial lower
bound on the nonzero principal angles between \(\cH_d(\Xin)\) and
\(\cH_d(\Xout)\), and admit efficient preparation of the initial harmonic
mixture?  Establishing both would place the instances in \(\BQP\)
by~\cite[Lemma~7]{lowe2026normalized}.
\end{enumerate}

Within its limits, the paper establishes a reusable criterion: a new clique
reduction obtains exact kernel multiplicity, an inverse-polynomial gap above
the whole kernel, and natural persistent ranks by certifying, for a finite
palette, local zero-weight kernels, private coercivity, exact filling, and
interfaces, with no growing-dimensional perturbation argument.

\bibliographystyle{plainurl}
\bibliography{references}

\clearpage
\appendix
\section{Proof of the all-chain transfer theorem}
\label{app:transfer-proof}

This appendix proves \cref{thm:transfer} from the finite local hypotheses. All constants below depend only on the fixed palette.

\subsection{Local scaling and harmonic padding}
\label{app:local-padding}

Let \(W_k\) be the diagonal matrix of simplex weights.  In normalized
coordinates,
\[
\partial_{w,k}=W_{k-1}^{-1}\partial_{1,k}W_k.
\]
For \(0<\lambda\le1\), invertible left and right scaling at the fixed rank of
\(\partial_{1,k}\) gives
\[
\sigma_{\min}^+(\partial_{w,k})
\ge\lambda^{k+1}\sigma_{\min}^+(\partial_{1,k}).
\]
Indeed, put \(K=\ker\partial_{1,k}\).  For every unit
\(x\perp W_k^{-1}K\),
\[
\dist(W_kx,K)\ge\sigma_{\min}(W_k).
\]
Apply the least nonzero singular value of the unweighted boundary.  The other
diagonal factor satisfies
\[
\sigma_{\min}(W_{k-1}^{-1})\ge1.
\]
The Hodge decomposition applies these estimates separately to the down- and
up-exact sectors and gives
\[
\gamma_+(\Delta_{Y_a,d_a}(\lambda))
\ge c_a\lambda^{2d_a+4}.
\]
A finite palette and \(d_a=2m_a-1\), \(m_a\le m\), give the common bound
\cref{eq:local-gap}.

For padding, use augmented reduced chain complexes, including degree \(-1\),
in every join bidegree:
\[
\widetilde C_d(Y_a*R_{\rm out})
=\bigoplus_{r+s=d-1}
\widetilde C_r(Y_a)\otimes\widetilde C_s(R_{\rm out}).
\]
On each summand the join Laplacian is the tensor sum
\[
\Delta_{Y_a*R_{\rm out}}
=\Delta_{Y_a}\otimes I+I\otimes\Delta_{R_{\rm out}}.
\]
The outside register has reduced homology only in its top degree and a
constant positive gap on every other mode.  At the global target degree,
therefore, the only outside-harmonic bidegree is the active target degree
tensored with \(\cH_{\rm out}\).  This proves the stated zero-weight kernel
and puts every other bidegree above a uniform floor.  Tensor the private
output projector with \(P_{\cH_{\rm out}}\); both outside differentials vanish
after this projection, preserving the private-pair bounds.

The same tensor-sum structure gives the padded form of
\cref{eq:local-gap}: the positive eigenvalues of the padded Laplacian in
degree \(d\) are sums \(\mu_1+\mu_2\) of factor eigenvalues.  In every
bidegree other than \((d_a,s_{\rm top})\) the outside eigenvalue \(\mu_2\) is
at least the constant gap of \(R_{\rm out}\), and in the bidegree
\((d_a,s_{\rm top})\) a positive sum has \(\mu_2=0\) and
\(\mu_1\ge\gamma_+(\Delta_{Y_a,d_a}(\lambda))\ge c\lambda^\kappa\).  Hence
\(\gamma_+\) of the padded degree-\(d\) Laplacian is at least
\(\min\{c\lambda^\kappa,\gamma_+(\Delta_{R_{\rm out}})\}\ge c'\lambda^\kappa\)
for \(0<\lambda\le1\).

Finally, the relative join formula identifies the padded connecting map with
the active connecting map tensored with
\(\cH_{\rm out}\).  Since all register weights are one, the positive-weight
diagonal gauge fixes the register coordinates.  Thus the padded boundary
intersection is \(\ran(\pi_a\otimes I_{\cH_{\rm out}})\), completing the proof
of \cref{lem:scaling-padding}.

\subsection{Global decomposition and interface energy}

If \(A=\varnothing\), then \(X_A=R\) and \(K_A=V\), and
\cref{eq:transfer-bound} holds with any \(C\ge\gamma_R^{-1}\): in its top
degree the register Laplacian \(\Delta_R=\partial_R^*\partial_R\) is the
tensor sum of the bowtie edge Laplacians, as in \cref{app:local-padding}, so
it has kernel \(V\) and a constant positive floor \(\gamma_R\) on
\(V^\perp\), whence
\(\|(I-P_V)x\|^2\le\gamma_R^{-1}\langle x,\Delta_Rx\rangle
\le\gamma_R^{-1}\langle x,\Delta_Rx\rangle/(g_A\lambda^\kappa)\)
because \(g_A\le1\) and \(\lambda\le1\).  Assume henceforth \(t_A\ge1\).

Write the target chain space as the orthogonal coordinate sum
\[
C_d(X_A)=C_d(R)\oplus\bigoplus_{j\in A}C_{d,j}^{\rm priv}.
\]
For a unit chain \(x\), write
\[
x=\omega_0+\sum_j\omega_j,
\qquad
x_j=\omega_0+\omega_j,
\qquad
U=\sum_j\|\omega_j\|^2,
\qquad
e=\langle x,\Delta_Ax\rangle.
\]
Let \(a=\partial_R\omega_0\), \(y_j=L_j\omega_j\), \(y=\sum_jy_j\), and \(r=a+y\). The vector \(r\) is the shared-register component of the global down-boundary. Since
\(
a+y_j=r-(y-y_j)
\),
\[
\sum_j\|a+y_j\|^2
\le2t_A\|r\|^2+2\sum_j\|y-y_j\|^2.
\]
The elementary identity
\[
\sum_j\|y-y_j\|^2
=(t_A-2)\|y\|^2+\sum_j\|y_j\|^2
\]
also holds for \(t_A=1\). Put
\[
\rho=\lambda^{-2}\|[L_1\ \cdots\ L_{t_A}]\|^2
\le N_*t_A,
\qquad
\chi=(t_A-2)_+\rho+N_*\le C t_A^2.
\]
Then~\cref{eq:interface} gives
\[
\sum_j\|y-y_j\|^2\le\chi\lambda^2U.
\]
All nonregister down-boundary coordinates and all up-boundary coordinates are private and partition by gadget. We obtain
\begin{equation}
\label{eq:appendix-local-energy}
\sum_j\|D_j(\lambda)x_j\|^2
\le2t_Ae+2\chi\lambda^2U.
\end{equation}

\subsection{Zero-weight and private-sector coercivity}

Let \(Q_j\) also denote the orthogonal projector onto the padded zero-weight private subspace \(Q_j\otimes\cH_{\rm out}\) of gadget \(j\), and let \(A_{j,0}=I-P_V-Q_j\) in the padded local chain space; by \cref{lem:scaling-padding} this is the projector onto the orthogonal complement of \(\ker D_{j,0}\). The positive singular-value floor of \(D_{j,0}\) on \(\ran A_{j,0}\), together with
\(
D_j(\lambda)=D_{j,0}+\lambda D_{j,1}
\), gives
\begin{align}
S_A:=\sum_j\|A_{j,0}x_j\|^2
&\le C\sum_j\|D_{j,0}x_j\|^2\notag\\
&\le C\left(
t_Ae+\chi\lambda^2U+
\lambda^2\sum_j\|x_j\|^2
\right)\notag\\
&\le C\left(t_Ae+\chi\lambda^2U+t_A\lambda^2\right).
\label{eq:appendix-SA}
\end{align}
The final inequality uses
\(
\sum_j\|x_j\|^2
=t_A\|\omega_0\|^2+U\le t_A
\).

Because \(T_jP_V=0\),
\[
T_jQ_jx_j=T_jx_j-T_jA_{j,0}x_j.
\]
The lower bound on \(T_j\) restricted to \(Q_j\), the estimate \(\|T_j\|\le C\lambda\), and orthogonality of the retained private output blocks imply
\begin{equation}
\label{eq:appendix-SQ}
S_Q:=\sum_j\|Q_jx_j\|^2
\le C\left(\frac{e}{\lambda^2}+S_A\right).
\end{equation}
The orthogonality here concerns the retained \(T_j\) outputs. The shared \(L_j\) outputs were handled by their horizontal operator norm and need not be orthogonal.

Let \(N=\|(I-P_V)x\|^2\). Summing the orthogonal decompositions \(P_V\oplus Q_j\oplus\ran A_{j,0}\) gives
\[
N\le U+t_A\|(I-P_V)\omega_0\|^2=S_A+S_Q.
\]
Since \(U\le N\), substituting~\cref{eq:appendix-SA,eq:appendix-SQ} and using \(\chi\le Ct_A^2\) yields
\[
N\le C\left[t_A\lambda^2+(t_A+\lambda^{-2})e\right]
+Ct_A^2\lambda^2N.
\]
Choosing \(\lambda\le c_0/t_{\max}\) absorbs the last term and proves~\cref{eq:geometric-leakage}.

\subsection{Logical coercivity}

For each padded gadget \(Y_j\), exact filling gives \(\ran\Pi_j\subseteq B_d(Y_j)\). The up-Laplacian \(\partial_{d+1}\partial_{d+1}^*\) of \(Y_j\) has kernel \(B_d(Y_j)^\perp\), and its positive eigenvalues are positive eigenvalues of the padded degree-\(d\) Laplacian, which the padded form of~\cref{eq:local-gap} bounds below by \(c\lambda^\kappa\). Therefore, after the common embedding into \(X_A\),
\[
\Delta_j^\uparrow\succeq c\lambda^\kappa\Pi_j.
\]
The top-dimensional register has no up-columns, and every degree-\((d+1)\) simplex with a private vertex belongs to exactly one gadget. Hence the global up-Laplacian is the exact sum of the embedded local up-Laplacians, proving~\cref{eq:logical-domination}.

Since \(K_A\subseteq V\), the squared distance decomposes as
\[
\|(I-P_{K_A})x\|^2
=N+\langle x,(P_V-P_{K_A})x\rangle.
\]
\Cref{eq:logical-gap,eq:logical-domination} give
\[
\langle x,(P_V-P_{K_A})x\rangle
\le\frac{e}{cg_A\lambda^\kappa}.
\]
Together with~\cref{eq:geometric-leakage},
\[
\|(I-P_{K_A})x\|^2
\le C\left[t_A\lambda^2+(t_A+\lambda^{-2})e
+\frac{e}{g_A\lambda^\kappa}\right].
\]
For \(\lambda\le c_0/t_{\max}\), \(g_A\le1\), and \(\kappa\ge2\), the last energy term dominates the middle energy terms up to a palette constant. This proves~\cref{eq:transfer-bound}.

\subsection{Exact multiplicity and the whole-positive-spectrum gap}

Let \(\mathcal L_A(E_A)\) be the span of all eigenvectors of \(\Delta_A\) with eigenvalue below \(E_A\). Every unit \(x\in\mathcal L_A(E_A)\) obeys \(e<E_A\). Choosing the constants in~\cref{eq:lambda-E-choice} makes the right-hand side of~\cref{eq:transfer-bound} strictly less than \(\eta^2<1\). Therefore \(P_{K_A}x\neq0\) for every nonzero \(x\in\mathcal L_A(E_A)\), so \(P_{K_A}|_{\mathcal L_A(E_A)}\) is injective and
\begin{equation}
\label{eq:low-dimension}
\dim\mathcal L_A(E_A)\le\dim K_A.
\end{equation}

Every \(v\in V\) remains a cycle after attaching the gadgets. By \cref{lem:independent-filling}, the map \(V/W_A\to H_d(X_A)\) is injective. Its domain has dimension \(\dim K_A\), so Hodge theory supplies at least \(\dim K_A\) zero eigenvectors of \(\Delta_A\). These zero vectors belong to \(\mathcal L_A(E_A)\). Combining this lower bound with~\cref{eq:low-dimension} forces equality throughout: the geometric kernel has dimension \(\dim K_A\), and no positive eigenvalue lies below \(E_A\).

If \(K_A=0\), any unit vector in \(\mathcal L_A(E_A)\) would have distance one from \(K_A\), contradicting the strict \(\eta<1\) bound. Thus the same argument proves \(\Delta_A\succeq E_AI\) in the zero-kernel case. No additional global relative-acyclicity assumption is required.

\subsection{Proof of quotient naturality}
\label{app:quotient-proof}

\Cref{eq:boundary-intersection} makes \(q_A\) well-defined and injective.
Its domain has dimension
\[
\dim(V/W_A)=\dim K_A.
\]
Hodge theory and
\cref{eq:kernel-gap-conclusion} give
\[
\dim H_d(X_A)=\dim\ker\Delta_A=\dim K_A,
\]
so \(q_A\) is also surjective.  If \(A\subseteq B\), both paths in the
diagram of \cref{thm:quotient} send \(v+W_A\) to the class of the same
register cycle \(v\) in \(X_B\).  Thus the diagram commutes.  The quotient
map is onto because \(W_A\subseteq W_B\), and its rank is
\(\dim(V/W_B)=\dim K_B\).

\section{History-gap and source calculations}
\label{app:source-details}

\subsection{Rational gain/loss split}
\label{app:rational-split}

On a clock bit \(c\) and work bit \(w\), define
\[
\begin{aligned}
u_0&=(\ket{00}+\ket{01}-\ket{10})/\sqrt3,&
u_1&=(\ket{00}-\ket{01}-\ket{11})/\sqrt3,\\
d_0&=(\ket{00}-\ket{10}-\ket{11})/\sqrt3,&
d_1&=(\ket{01}-\ket{10}+\ket{11})/\sqrt3.
\end{aligned}
\]
The vectors in each pair are orthonormal.  If consecutive history amplitudes
are \(f_t,f_{t+1}\), orthogonality to \(u_0,u_1\) is equivalent to
\(f_{t+1}=\sqrt2Hf_t\); orthogonality to \(d_0,d_1\) is equivalent to
\(f_{t+1}=Hf_t/\sqrt2\).  Applying the first relation to one work bit and the
second to the other gives \(H\otimes H\).

\subsection{Weighted path estimate}

With the unary encoding \(\ket t=\ket{1^t0^{T-t}}\), the legal clock span and
its orthogonal complement are invariant under the neighboring-clock-guarded
propagation, input, and output projectors.  The diagonal clock term
\[
H_{\rm clock}=\sum_{j=1}^{T-1}\ket{01}\!\bra{01}_{j,j+1}
\]
has energy at least one on every illegal computational-basis clock string.
By Hermiticity, the whole unary-clock Hilbert space is the orthogonal sum of
the invariant legal and illegal sectors.  It therefore
suffices to prove the smaller inverse-polynomial bound on the legal sector.
For a legal-clock state, write \(f_t=s_tV_t\xi_t\), with the notation of~\cref{eq:history-isometry}. Its norm and initial-history energy are
\[
N(\xi)=\sum_tw_t\|\xi_t\|^2,
\qquad
E(\xi)=\sum_{t=1}^{L-1}c_t\|\xi_t-\xi_{t-1}\|^2
+\langle\xi_0,A_{\rm in}\xi_0\rangle,
\]
where \(w_t=s_t^2\in[1,2]\), \(c_t\in\{1/2,2/3\}\), and
\(A_{\rm in}\) is the clean-input penalty at time zero, a sum of commuting
computational-basis projectors, one per clean qubit.  It is not in general a
projector: a basis string violating \(k\) clean bits has energy \(k\).  All
that is used below is
\[
\ker A_{\rm in}=C,\qquad A_{\rm in}\succeq P_{C^\perp}.
\]
Split every \(\xi_t\) into the clean subspace \(C\) and its orthogonal
complement, and split the energy accordingly as \(E(\xi)=E_C+E_B\) with
\[
E_C=\sum_{t=1}^{L-1}c_t\|\xi_t^C-\xi_{t-1}^C\|^2,\qquad
E_B=\sum_{t=1}^{L-1}c_t\|\xi_t^{C^\perp}-\xi_{t-1}^{C^\perp}\|^2
+\langle\xi_0,A_{\rm in}\xi_0\rangle ;
\]
the two difference terms are orthogonal because \(C\) is invariant under
the splitting.

For a state orthogonal to all valid histories, the clean component has weighted mean zero:
\[
\sum_tw_t\xi_t^C=0.
\]
The weighted variance identity and path Cauchy--Schwarz inequality give
\[
\sum_tw_t\|\xi_t^C\|^2
=\frac1{2\sum_tw_t}
\sum_{i,j}w_iw_j\|\xi_i^C-\xi_j^C\|^2
\le2L^2E_C.
\]
For the dirty component, set
\(a=\|\xi_0^{C^\perp}\|^2\le\langle\xi_0,A_{\rm in}\xi_0\rangle\le E_B\) and
\(
D_B=\sum_t\|\xi_t^{C^\perp}-\xi_{t-1}^{C^\perp}\|^2\le2E_B
\), using \(c_t\ge1/2\).
The bound
\[
\|\xi_t^{C^\perp}\|^2\le2a+2(L-1)D_B
\]
and \(\sum_tw_t\le2L\) imply
\[
\sum_tw_t\|\xi_t^{C^\perp}\|^2
\le2L\bigl(2E_B+4(L-1)E_B\bigr)\le8L^2E_B.
\]
Adding both sectors proves \(N(\xi)\le8L^2E(\xi)\) on the orthogonal complement of the complete initial kernel, establishing the first bound in~\cref{eq:history-gaps}.

\subsection{Adding the output projector}

We record the two-positive-operator estimate used for the second bound in~\cref{eq:history-gaps}. Let \(P\) project onto \(\ker A\), suppose
\(A\succeq g(I-P)\), and let \(R\) be an orthogonal projector.  Put
\(\mathcal S=\ker A\cap\ker R\), and suppose \(PRP\succeq\alpha P\) on
\(\ker A\ominus\mathcal S\).  On every nontrivial canonical block for the
pair \(P,R\), for some \(\rho\in[\alpha,1]\),
\[
 P=\begin{pmatrix}1&0\\0&0\end{pmatrix},\qquad
 R=\begin{pmatrix}
 \rho&\sqrt{\rho(1-\rho)}\\
 \sqrt{\rho(1-\rho)}&1-\rho
 \end{pmatrix}.
\]
Thus \(g(I-P)+R\) has trace \(g+1\), determinant \(g\rho\), and
smallest eigenvalue
\[
 \frac{g+1-\sqrt{(g+1)^2-4g\rho}}2.
\]
This quantity increases with \(\rho\); the one-dimensional canonical blocks
outside \(\mathcal S\) obey the same lower bound.  Since
\(A+R\succeq g(I-P)+R\) and both operators have kernel \(\mathcal S\),
\begin{equation}
\label{eq:two-positive}
\gamma_+(A+R)
\ge
\frac{g+1-\sqrt{(g+1)^2-4g\alpha}}2
\ge\frac{g\alpha}{g+1}.
\end{equation}
For the history Hamiltonian, \(g=1/(8L^2)\) and \(\alpha=(1-r)/Z\ge1/(3L)\) when \(r\le1/3\). Substitution into~\cref{eq:two-positive} gives
\[
\gamma_+(H_{\rm out})
\ge\frac1{3L(8L^2+1)}.
\]
The kernel identity follows directly from positivity:
\[
\ker(H_{\rm in}+H_{\rm rej})
=\ker H_{\rm in}\cap\ker H_{\rm rej}
=\cV\ker(I-M).
\]

\subsection{Exact eight-label compression}

Let \(U_x\) denote the fixed-input verifier and \(q\) its output bit. The label \(z\) is never changed. After computing and uncomputing the constant-size acceptance predicate, compression through all clean work registers has no matrix element between distinct labels. On \(z=000\), acceptance is the identity. On five labels it is the scalar
\(
p_x=\|P_{\rm acc}U_x\ket{0\cdots0}\|^2
\),
and on two labels it is zero. This proves~\cref{eq:eight-label} even for a coherent vector across label values.

The normalized trace is
\[
\frac{\operatorname{Tr}M_x}{8}
=\frac{1+5p_x}{8}.
\]
If \(p_x=1\), this equals \(3/4\) and coincides with the perfect-space fraction. If \(p_x\le1/3\), it is at most \(1/3\), the perfect-space fraction is exactly \(1/8\), and every off-perfect eigenvalue is at most \(1/3\). Thus the construction simultaneously supplies the trace promise, the perfect-space separation, and the off-perfect ceiling needed by the output-gap proof.

\subsection{Source-interface boundary}
\label{app:source-boundaries}

The exact source used by \cref{thm:weighted-hardness} begins with a
fixed-input \(\BQPone^{G_2}\) verifier.  The three mixed label bits are added
by the reduction and are the entire input space of
\cref{eq:eight-label}; the verifier's input and all predicate scratch remain
clean.  Consequently the initial history dimension is exactly eight.  No
trace-to-rank inference, unrestricted \(\SDQCone\) identification, or
approximate gate compilation enters the proof.

The finite gadget evidence needed to compile these circuit terms is given in
\cref{app:palette}.  The theorem uses those exact finite objects and does not
infer their topology from floating-point spectra.

For an arbitrary acceptance operator with
\(M|_{S^\perp}\preceq rI\), put \(f=\dim S/D\) and
\(p=\operatorname{Tr}(M)/D\).  Spectral decomposition gives
\[
f\le p\le r+(1-r)f.
\]
Thus trace promises \(p\ge a\) versus \(p\le b\) force a perfect-space
fraction gap only when
\(\max\{0,(a-r)/(1-r)\}>b\).  The eight-label source avoids this obstruction
by fixing the perfect fractions directly.

An exact circuit over \(\{X,\CX,\CCX,H\}\) admits the following real-gate
extension.  Add one maximally mixed, unmeasured spectator \(a\) and replace
each \(H_j\) by \(H_j\otimes H_a\).  If there are \(h\) Hadamards, then
\[
U'=U\otimes H_a^h,\qquad J'=J\otimes I_a,\qquad
P'_{\rm acc}=P_{\rm acc}\otimes I_a.
\]
For either parity of \(h\),
\[
(J')^*(U')^*P'_{\rm acc}U'J'=M\otimes I_a.
\]
Both the input dimension and perfect eigenspace dimension double, while the
fraction and off-perfect ceiling are unchanged.  This calculation does not
cover complex phases.

\section{Finite palette and guard certificates}
\label{app:palette}

This appendix identifies the finite objects used in \cref{lem:palette}.  The
underlying active graphs are four explicit constructors in the source artifact:
\path{state_0m1}, \path{state_00m11}, \path{state_01m10}, and
\path{state_00m10m11}.  They accompany Rudolph's gain/loss
construction~\cite{rudolph2026universal}.
The supplementary certificate archive records the complete graphs, oriented
integer fillings, kernel bases, modular rank witnesses, and Gram determinants.
No numerical eigenvalue threshold is used.

\begin{table}[ht]
\centering
\begin{tabularx}{\linewidth}{@{}>{\raggedright\arraybackslash}Xrrrr>{\raggedright\arraybackslash}X@{}}
\toprule
Active vector & degree & vertices & \(\dim V\) & \(\dim Q\) & certificate mod \(p\) \\
\midrule
\(\ket0-\ket1\) & 1 & 16 & 2 & 8 &
\(\det G=512\) \\
\(\ket{00}-\ket{11}\) & 3 & 33 & 4 & 112 &
\(\det G=655970\) \\
\(\ket{01}-\ket{10}\) & 3 & 33 & 4 & 112 &
signed relabeling \\
\(\ket{00}-\ket{10}-\ket{11}\) & 3 & 41 & 4 & 176 &
\(\det G=797443\) \\
three remaining gain/loss vectors & 3 & 41 & 4 & 176 &
signed relabelings \\
\bottomrule
\end{tabularx}
\caption{Exact active-palette data.  Vertex counts include the \(7m_a\)
register vertices of the active bowties.  In each determinant row,
\(p=1000003\) and \(G=T_{a,1}^*T_{a,1}\) on \(Q_a\); a nonzero residue
proves injectivity over \(\mathbb Q\).}
\label{tab:palette-certificates}
\end{table}

For the one-qubit atom, exact boundary ranks are \(15,24\) and the Betti
numbers are \((1,1,0)\); a denominator-one 24-term two-chain fills the target
cycle.  At zero weight the target differential pair has rank \(30\), and the
two register cycles plus eight private vectors give its complete
ten-dimensional kernel.  For \(\ket{00}-\ket{11}\), the clique counts in
degrees zero through four are
\((33,241,646,712,272)\), the boundary ranks are
\((32,209,437,272)\), and the Betti numbers are \((1,0,0,3,0)\).
A denominator-one 248-term four-chain fills the target; rank \(596\) of the
zero-weight pair matches the four register and 112 private kernel vectors.
Swapping the second bowtie petals transports this certificate exactly to
\(\ket{01}-\ket{10}\).

For the three-term representative, the clique counts in degrees zero through
five are
\[
(41,322,914,1082,477,30),
\]
and the boundary ranks are \((40,282,632,447,30)\).  Hence the target
degree-three Betti number is three.  A denominator-one 382-term four-chain
fills \(\ket{00}-\ket{10}-\ket{11}\), while three complementary register
cycles remain independent modulo boundaries.  The zero-weight pair has rank
\(902\), matching its kernel \(V\oplus Q\) of dimension \(4+176\).
Signed permutations of bowtie petals transport the graph, filling, kernel,
and private pair to the other three gain/loss vectors.

\subsection{Selected-cycle guards}

We record the guard closure because it is part of the theorem interface, not
an appeal to a general integer-state gadget.  Let \(Y\) implement
\(\pi_\phi\) on register \(R\), let \(B\) be a fresh bowtie with petals
\(S_0,S_1\), and choose \(\beta\in\{0,1\}\).  Add all edges of \(B\), join
every vertex of \(R\) to every vertex of \(B\), and join each private vertex
of \(Y\) only to the selected petal \(S_\beta\).  The resulting flag complex
is
\begin{equation}
\label{eq:guard-union}
Z=(Y*S_\beta)\cup(R*B),\qquad
(Y*S_\beta)\cap(R*B)=R*S_\beta .
\end{equation}
Its new register is \(R'=R*B\).  If \(\gamma_\beta\) is the normalized
oriented cycle of the selected petal, then \(Z\) implements
\[
\pi_\phi\otimes\ket{\beta}\!\bra{\beta},\qquad
Q'=Q\otimes\mathbb C\gamma_\beta .
\]

Indeed, reduced relative chains give
\[
\widetilde C_*(Z,R')
\cong\widetilde C_*(Y,R)\widehat\otimes
\widetilde C_*(S_\beta).
\]
The connecting map therefore tensors the old exact boundary intersection
with \([\gamma_\beta]\), proving exact filling and no additional killed
register class.  At zero private weight the private differential is the
corresponding tensor-sum differential.  Since \(S_\beta\) has reduced
homology only in degree one, the complete target kernel is
\(V'\oplus Q'\); every other bidegree pays a fixed positive guard gap.
Finally, project the old private output pair onto the selected harmonic line.
Both outside differentials vanish under this projection, so
\[
T_Z(\lambda)=\lambda T_{Z,1},\qquad
T_{Z,1}V'=0,\qquad
\|T_{Z,1}(q\otimes\gamma_\beta)\|
\ge b\|q\otimes\gamma_\beta\|.
\]
Removing the unique old private vertex also gives
\(LL^*\preceq N_{\rm priv}\lambda^2I\).  Iterating this construction for the
constant number of required basis guards preserves a fixed finite palette.

The basis projector itself is the special case
\(R\cup(v*S_\beta)\), where a private apex \(v\) cones the selected petal.
Its private zero-weight kernel \(Q\) is zero, so the private-pair condition is
vacuous.  Together with \cref{tab:palette-certificates}, this proves
\cref{lem:palette}.

\subsection{Reproducibility boundary}

The supplementary archive contains the exact witnesses in
\begin{itemize}
\item \path{RUDOLPH_REPRESENTATIVE_BULK_CERTIFICATE.json},
\item \path{REMAINING_ACTIVE_ATOM_CERTIFICATES.json},
\item \path{ACTIVE_HADAMARD_ORBIT_CERTIFICATE.json}, and
\item \path{SELECTED_CYCLE_GUARD_CHECKS.json}.
\end{itemize}
Their checkers recompute the stated ranks, fillings, relabelings, and guard
identities.  These finite checks certify the displayed palette; the analytic
implication from those certificates is the proof of \cref{thm:transfer}.

\section*{AI Disclosure}
We used ChatGPT and Claude as auxiliary tools for literature search,
brainstorming, proof checking, and drafting and editing of the exposition
(Sections~1, 3--8, and Appendices~A--C).
The authors evaluated and revised all AI-assisted material and take full
responsibility for the correctness and originality of the submission,
including every proof and reference.

\end{document}